\documentclass[letterpaper, 10 pt, conference]{ieeecolor}
\usepackage{cite}
\usepackage{amsmath,amssymb,amsfonts}
\usepackage{algorithmic}
\usepackage{graphicx}
\usepackage{algorithm,algorithmic}
\usepackage{hyperref}
\hypersetup{hidelinks=true}
\usepackage{textcomp}
\usepackage{enumerate}
\usepackage{dsfont}
\usepackage{xcolor}
\usepackage{comment}
\usepackage[normalem]{ulem}
\usepackage{accents}
\newcommand{\ubar}[1]{\underaccent{\bar}{#1}}

\newtheorem{theorem}{Theorem}
\newtheorem{lemma}{Lemma}
\newtheorem{remark}{Remark}
\newtheorem{definition}{Definition}
\newtheorem{assumption}{Assumption}
\newtheorem{proposition}{Proposition}
\newtheorem{example}{Example}
\newtheorem{corollary}{Corollary}
\newtheorem{problem}{Problem}
\let\oldtheorem\theorem
\let\endoldtheorem\endtheorem
\def\theorem{\begingroup \oldtheorem }
\def\endtheorem{ \hfill $\square$\endoldtheorem \endgroup}

\let\olddefinition\definition
\let\oldenddefinition\enddefinition
\def\definition{\begingroup \olddefinition  }
\def\enddefinition{ \hfill $\square$\oldenddefinition\endgroup}

\let\oldexample\example
\let\endoldexample\endexample
\def\example{\begingroup \oldexample }
\def\endexample{ \hfill $\square$\endoldexample \endgroup}

\DeclareMathOperator{\col}{col}

\usepackage[prependcaption,colorinlistoftodos]{todonotes}

\title{Stability of Information-Based Routing in Dynamic Transportation Networks}
\author{Shaya Garjani, Ashish Cherukuri, Bayu Jayawardhana, and Nima Monshizadeh
\thanks{The authors are with the Engineering and Technology Institute Groningen and J.C. Willems Center for Systems and Control, Faculty of Science and Engineering, University of Groningen, The Netherlands (email: \href{s.garjani@rug.nl}{s.garjani@rug.nl}, \href{a.k.cherukuri@rug.nl}{a.k.cherukuri@rug.nl}, \href{b.jayawardhana@rug.nl}{b.jayawardhana@rug.nl}, \href{n.monshizadeh@rug.nl}{n.monshizadeh@rug.nl})}}
\IEEEoverridecommandlockouts
\begin{document}

\maketitle

\begin{abstract}
Recent studies on transportation networks have shown that real-time route guidance can inadvertently induce congestion or oscillatory traffic patterns. Nevertheless, such technologies also offer a promising opportunity to manage traffic non-intrusively by shaping the information delivered to users, thereby mitigating congestion and enhancing network stability.
A key step toward this goal is to identify information signals that ensure the existence of an equilibrium with desirable stability and convergence properties. This challenge is particularly relevant when traffic density and routing dynamics evolve concurrently, as increasingly occurs with digital signaling and real-time navigation technologies.
To address this, we analyze a parallel-path transportation network with a single origin-destination pair, incorporating joint traffic density and logit-based routing dynamics that evolve at the same timescale. 
We characterize a class of density-dependent traffic information that guarantees a unique equilibrium in the free-flow regime, ensures its asymptotic stability, and keeps traffic densities within the free-flow region for all time. The theoretical results are complemented by a numerical case study demonstrating how the framework can inform the design of traffic information that reduces total travel time without compromising credibility.

\end{abstract}

\section{Introduction}\label{sec1}
For the past decade, there has been a growing interest in controlling congestion in dynamical transportation networks. This trend is largely due to rapid advances
in route-recommendation technologies that allow drivers to continuously adjust their routes based on real-time traffic information. While these technologies enable individuals to reduce their personal travel times, their widespread use can lead to inefficiencies or even instability in traffic patterns \cite{acemoglu2018informational,toso2023modeling,10520878}. Consequently, analyzing traffic dynamics under adaptive driver behavior has become an essential research problem.

Transportation systems have been extensively studied within the context of routing games \cite{wardrop1952road, beckmann1956studies,roughgarden2002bad}, in which drivers select routes that minimize their travel time. The equilibrium of this game, referred to as the Wardrop equilibrium \cite{wardrop1952road}, corresponds to a condition where no driver can unilaterally reduce their travel time by changing routes.  Although classical routing games provide a useful framework for characterizing the equilibria and assessing their efficiency \cite{roughgarden2002bad}, they neglect the dynamical aspects of the network, offering no insight into how routing patterns converge to these equilibria. To address this, evolutionary dynamics \cite{hofbauer1998evolutionary, sandholm2010population} have been proposed, which describe the evolution of routing decisions. Under standard monotonicity assumptions, routing games admit a concave potential function \cite{beckmann1956studies} and under a perturbed best-response dynamic, such as the logit dynamic, they admit a globally asymptotically stable equilibrium\cite{sandholm2010population}. 

Another dynamical aspect of transportation networks, that has been overlooked in the classical routing game, is the underlying traffic flow dynamics, which, in macroscopic frameworks, are often modeled as mass-conservation laws \cite{siri2021freeway}. The stability of dynamic traffic flow networks has been investigated in \cite{toso2023modeling,lovisari2014stability,nilsson2023strong,toso2023impact}. However, these studies consider static routing preferences that are either fixed \cite{lovisari2014stability,nilsson2023strong} or  instantaneously equilibrated based on the current traffic state \cite{toso2023modeling,toso2023impact}.

In emerging transportation systems, traffic dynamics and routing decisions are becoming increasingly intertwined. The analysis of traffic networks with coupled density–routing dynamics has been addressed in only a handful of studies in the literature \cite{como2013stability, 10520878}. In \cite{como2013stability}, the authors consider a directed acyclic graph with macroscopic density dynamics and a perturbed best-response dynamics governing the evolution of routing preferences. They establish the asymptotic stability of the system to the unique Wardrop equilibrium using a singular perturbation approach. The alignment of the equilibrium with the social optimum was later investigated in \cite{como2021distributed} by introducing decentralized monotone feedback tolls. This work also extended the results of \cite{como2013stability} to cyclic multigraphs.

Most of these works, however, rely on simplifying assumptions such as separated time scales between density and routing evolution or strictly monotone feedback mappings, which limit their applicability in modern settings where routing and traffic states interact simultaneously through 
real-time information platforms. Modern navigation platforms continuously update route recommendations in response to evolving traffic, while the resulting routing choices, in turn, reshape these states at comparable timescales. Consequently, understanding and ensuring the stability of such joint density–routing dynamics asks for concurrent analysis of both processes rather than treating them separately or assuming one evolves much slower than the other. The work in \cite{10520878} proposes a coupled density-routing dynamics where the routing evolve according to the replicator dynamics \cite{hofbauer1998evolutionary}, in which users imitate previous travelers' preferences. While the replicator formulation offers valuable analytical tractability, its behavioral interpretation in the context of real-time traffic networks is less direct as drivers typically do not observe other travelers’ realized travel times or adapt by imitation. 

As mentioned, the information communicated to users plays a crucial role in the collective outcomes, and providing raw travel delay information can amplify fluctuations or even destabilize the network, whereas suitably designed information signals can steer behavior toward more efficient and stable operating conditions. This motivates an information-design perspective, where the goal is to identify families of information functions 
that promote desirable equilibria while preserving information credibility and user responsiveness.

\textit{Contribution:}
In this work, we take a preparatory step toward such an information-design framework by characterizing a flexible class of density-dependent information signals that guarantees the existence of an equilibrium in the free-flow regime and ensures its asymptotic stability. The analysis explicitly captures the coupling between density and routing dynamics evolving at the same timescale. In contrast to most existing stability results that rely on strictly monotone information mappings, a hallmark assumption in earlier studies, our framework relaxes this requirement to provide greater flexibility for designing or optimizing information signals in practice. This theoretical groundwork enables subsequent research to incorporate system-level optimization objectives, e.g., minimizing total travel time under the constraint that information signals belong to the identified admissible class. We illustrate this potential through a numerical case study demonstrating how such signals can reduce congestion without compromising credibility.

The rest of the paper is organized as follows. Section \ref{problem_formulation_sec} introduces the coupled density-routing dynamics and the problem formulation.
Section \ref{condition_sec} specifies a class of traffic information ensuring the existence of a unique asymptotically stable free-flow equilibrium. Section \ref{example} presents the information design framework and numerical results. Section \ref{conclusion_sec} presents the concluding remarks.

\textit{Notation:} We denote the set of real and nonnegative real numbers by $\mathbb{R}$ and $\mathbb{R}_{\geq 0}$, respectively. The standard Euclidean norm is denoted by $\|\cdot\|$. Let $\mathds{1}$ denote the vector of all ones with dimensions deduced from the context. Given vectors $x_1, \ldots, x_n$, we denote $\col(x_j)=[x_1^\top, \ldots, x_n^\top]^\top$. For vectors $x\in\mathbb{R}^n$ and $y\in\mathbb{R}^n$, $x\leq y$ denotes the element-wise inequalities $x_j\leq y_j$, $ j=1,\ldots, n$. For any set $\mathcal{Y}\subseteq \mathbb{R}^n$, $\text{int}(\mathcal{Y})$ denotes its interior.

\section{Problem Statement}\label{problem_formulation_sec}
In this section, the dynamic model of the transportation network is introduced. The traffic density dynamics and routing dynamics are described, and the problem is formulated.
\subsection{Traffic density-flow model}
Consider $p$ parallel paths connecting 
an origin and a destination. 
For each path $j \in \mathcal{P} := \{1, \ldots, p\}$, the macroscopic density dynamics are modeled as a mass conservation law \cite{daganzo1994cell} given by
\begin{equation}\label{density_dyn}
    \dot{x}_j(t)=-f_j(x_j(t))+\lambda r_j(t),
\end{equation}
where $x_j$ denotes the traffic density on path $j$, $\lambda>0$ is the constant incoming flow, $f_j:\mathbb{R}_{\geq 0}\rightarrow \mathbb{R}_{\geq 0}$ denotes the function describing the outflow from path $j$, and $r_j:\mathbb{R}_{\geq 0}\rightarrow[0, 1]$ indicates the portion of drivers opting for path $j$ at any given time $t$, whose dynamics will be specified later. Each path $j\in\mathcal{P}$ has a critical density $\bar{x}_j$ and critical flow $\bar{f}_j:=f_j(\bar{x}_j)$. The region $\mathcal{X}_j:=[0,\bar{x}_j]$ represents the free-flow regime on path $j\in\mathcal{P}$, while densities outside this region correspond to the congested regime. We make the following assumption on the outflow function $f_j$.
\begin{assumption}\label{outflow_ass}
    The map $f_j$ is Lipschitz continuous and strictly monotone on $\mathcal{X}_j$ 
    , i.e.
    \begin{equation*}
         (x_j-x_j')(f_j(x_j)-f_j(x_j'))>0, \quad \forall x_j, x_j' \in \mathcal{X}_j,
    \end{equation*}
    and $f_j(0)=0$.
\end{assumption}

    A few  notable instances of outflow functions $f_j$ satisfying Assumption \ref{outflow_ass} 
     are provided below: 
    \begin{description}
\item[(i).]  $f_j(x_j) = (2\bar{f}_j/\bar{x}_j)x_j(1-(x_j/2\bar{x}_j))$ known as Greenshield's fundamental diagram  \cite{greenshields1935study}; 
\item[(ii).] $f_j(x_j)=\min{\{(\bar{f}_j/\bar{x}_j)x_j,\ \bar{f}_j+w_j(\bar{x}_j-x_j)\}}$ with $w_j>0$, known as Daganzo's triangular fundamental relation \cite{daganzo1995cell}; and
    \item[(iii).] $f_j=\bar{f}_j(1-e^{-\theta_j x_j}),\, \theta_j>0$ \cite{como2021distributed}.
    \end{description}

\subsection{Information-based traffic dynamics}

We assume that for all $ j\in\mathcal{P}$, the routing ratios $r_j$ in \eqref{density_dyn} evolve according to the following logit dynamics
\begin{equation*}
    \dot{r}_j(t) = -r_j(t) + \frac{e^{-\eta u_j(x_j(t))}}{\sum_{i=1}^p{e^{-\eta u_i(x_i(t))}}},
\end{equation*}
where $u_j(t)\in \mathbb{R}_{\geq 0}$ is an input signal representing communicated traffic information about path $j$ and is modeled as a continuous function of density, i.e., $u_j(t) \equiv u_j(x_j(t))$. 

The logit dynamics, also known as the perturbed best-response \cite{sandholm2010population}, captures the adaptive behavior of a population of drivers who tend to select routes with lower perceived travel times, while accounting for limited responsiveness to the provided information. The parameter $\eta\geq0$ reflects the population’s responsiveness; with higher values leading to route choices more concentrated on paths with lower perceived travel times, and lower values resulting in more dispersed choices. In the extreme case $\eta \to 0$, drivers ignore the information entirely and choose routes uniformly, whereas for $\eta \to \infty$, the logit dynamics converge to the classical best-response dynamics where all drivers choose the routes with minimal perceived travel times.

The resulting joint density-routing dynamics are expressed as 
\begin{subequations}\label{vec_dyn_eq}
    \begin{align}
        &\dot{x}(t) = -f(x(t)) + \lambda r(t),\label{vec_dyn_dens} \\ 
        &\dot{r}(t) = -r(t) + \frac{e^{-\eta u(x)}}{\mathds{1}^{\top} e^{-\eta u(x)}},\label{vec_dyn_ratio}
    \end{align}
\end{subequations}

where $x:=\col(x_j)$, $f:=\col(f_j)$, and $r:=\col(r_j)$.

Our goal here is to investigate how the routing decisions can be influenced in real-time by {traffic information} so that the system remains in the free-flow regime. Specifically, we determine a class of feedback information {signals} such that the above dynamics admit an equilibrium in the free-flow region $\mathcal{X}:=\mathcal{X}_1\times\ldots\times\mathcal{X}_p$, and the  traffic densities and routing ratios asymptotically converge to it.
This goal is formalized below:

\begin{problem}\label{problem}
Specify a class of feedback information \mbox{signals} $u$, denoted by $\mathcal{U}$, such that for any $u\in\mathcal{U}$, the system \eqref{vec_dyn_eq} admits a (unique)  asymptotically stable
equilibrium $(x^\mathrm{eq}, r^\mathrm{eq})\in \mathrm{int}(\mathcal{X})\times \mathcal{R}$, where $\mathcal{R}:=\{r\in \mathbb{R}_{\geq 0}^p\ |\ \mathds{1}^\top r =1\}.$
\end{problem}

\section{Equilibrium and Stability Analysis}\label{condition_sec}
\subsection{Existence of equilibrium}

In this subsection, we establish some sufficient conditions under which \eqref{vec_dyn_eq} admits an equilibrium in $\mathcal{X}\times \mathcal{R}$.

The set of equilibria of \eqref{vec_dyn_eq} is given by the points $(x^{\mathrm{eq}}, r^{\mathrm{eq}})$ satisfying
\begin{equation}\label{eq_eq}
    0 = -f(x^{\mathrm{eq}})+r^{\mathrm{eq}}\lambda, \quad r^{\mathrm{eq}}=\frac{e^{-\eta u(x^{\mathrm{eq}})}}{\mathds{1}^\top e^{-\eta u(x^{\mathrm{eq}})}}.
\end{equation}
Hence, \eqref{vec_dyn_eq} admits an equilibrium in the free-flow region $\mathcal{X}$, if and only if 
\begin{equation}\label{nec_suff_exist_eq}
    f(x) = \lambda \frac{e^{-\eta u(x)}}{\mathds{1}^\top e^{-\eta u(x)}}
\end{equation}
has a solution $x^{\mathrm{eq}}\in\mathcal{X}$.

Since $u_j$ is continuous, it admits a lower and upper bound in $\mathcal{X}_j$, namely $u_j\in [\ubar{u}_j, \bar{u}_j]$. We can leverage the boundedness of $u_j$ to establish the following result.

\begin{proposition}\label{prop_existence_suff}
    Let Assumption \ref{outflow_ass} hold. Then, \eqref{vec_dyn_eq} admits an equilibrium in $\mathcal{X}\times \mathcal{R}$ 
    if   
    \begin{equation}\label{exist_cond}
    \lambda \leq \bar{f}_je^{\eta \ubar{u}_j}\sum_{i\in \mathcal{P}}{e^{-\eta\bar{u}_i}}, \quad \forall j\in \mathcal{P}.
    \end{equation}
\end{proposition}

\begin{proof}
    By $u_j\in [\ubar{u}_j, \bar{u}_j]$, we have that
    
\begin{equation}\label{bounds_exists}
        0 \leq \lambda \frac{e^{-\eta u(x)}}{\mathds{1}^\top e^{-\eta u(x)}}\leq \lambda\frac{e^{-\eta \ubar{u}}}{\mathds{1}^\top e^{-\eta \bar{u}}}, \quad \forall x\in\mathcal{X}.
    \end{equation}
    Denote $\bar f:=\col(\bar f_j)$, $\bar{u}:=\col(\bar u_j)$, and $\ubar{u}=\col(\ubar u_j)$. By \eqref{exist_cond}, it follows that 
    $
\lambda{e^{-\eta \ubar{u}}}\leq ({\mathds{1}^\top e^{-\eta \bar{u}}})\bar f$. Combining the latter inequality with \eqref{bounds_exists}, we find that
    \[
    0 \leq \lambda \frac{e^{-\eta u(x)}}{\mathds{1}^\top e^{-\eta u(x)}}\leq 
    \bar f
    , \quad \forall x\in\mathcal{X}.
    \]
By strict monotonicity of $f_j$ (Assumption~\ref{outflow_ass}), we obtain 
    \begin{equation*}
        0\leq \phi(x) \leq \bar{x}, \quad \phi(x):= f^{-1}\left(\lambda\frac{e^{-\eta u(x)}}{\mathds{1}^\top e^{-\eta u(x)}}\right).
    \end{equation*}
Noting that $\phi: \mathcal{X} \rightarrow \mathcal{X}$ is a continious mapping into itself, we can apply Brouwer’s fixed-point theorem \cite[Theorem 9.2]{pata2019fixed}   to conclude that there exists a point $x^\mathrm{eq}$ such that $\phi(x^\mathrm{eq})=x^\mathrm{eq}$.
Consequently, $x^\mathrm{eq}$ satisfies \eqref{nec_suff_exist_eq}, which completes the proof.
\end{proof}

\begin{remark}
    From \eqref{nec_suff_exist_eq} and $u_j\in [\ubar{u}_j, \bar{u}_j]$, it can be shown that a necessary condition for the existence of an equilibrium in the free-flow region is $\lambda \leq \bar{f}_je^{\eta \bar{u}_j}\sum_{i\in \mathcal{P}}{e^{-\eta\ubar{u}_i}}, \,\forall j\in \mathcal{P}$. Note that, this condition is less restrictive than \eqref{exist_cond}. As the term $\eta (\bar{u}_i-\ubar{u}_j)$ for $i,j\in\mathcal{P}$ decreases, the gap between this necessary condition and the sufficient condition \eqref{exist_cond} shrinks.
\end{remark}
\subsection{Uniqueness of equilibrium}

In this subsection, we discuss conditions for guaranteeing 
uniqueness of the equilibrium.
To this end, we first impose a slightly stronger assumption than Assumption \ref{outflow_ass}, where the strict monotonicity property of outflow functions is replaced by strong monotonicity.
\begin{assumption}\label{f_strong_mono}
     Each outflow function $f_j$ is Lipschitz continuous, and $\mu_j$-strongly monotone in $\mathcal{X}_j$, i.e., for each $j\in \mathcal{P}$, there exists $ \mu_j>0$ such that
    \begin{equation*}
         (x_j-x_j')(f_j(x_j)-f_j(x_j'))\geq \mu_j(x_j-x_j')^2, \quad \forall x_j, x_j' \in \mathcal{X}_j,
    \end{equation*}
    and $f_j(0)=0$.
\end{assumption}

Furthermore, we impose the following regularity condition on the traffic information map  $u_j$.
\begin{assumption}\label{lipschitz_u}
    For each $j\in \mathcal{P}$, the map $u_j$ is continuously differentiable on $\mathcal{X}_j$, i.e., for all $j\in\mathcal{P}$, there exists $\ell_j\geq 0$ such that
    \[
    \left|\frac{\partial u_j}{\partial x_j}(x_j)\right| \leq \ell_j, \quad \forall x_j\in\mathcal{X}_j. 
    \]
\end{assumption}

We require the following technical lemma before stating the main result of this subsection. 

\begin{lemma}\label{softmax_lipschitz}
    The softmax function
    \begin{equation}\label{softmax_func}
        \sigma(z) :=\frac{e^{-\eta z}}{\mathds{1}^\top e^{-\eta z}}
    \end{equation}
    is globally $\frac{\eta}{2}$-Lipschitz continuous.
\end{lemma}
\begin{proof}
     The Jacobian of $\sigma(z)$ is
     given by
    \begin{equation*}
        J(z) = -\eta (\text{diag}(\sigma(z)) - \sigma(z)\sigma(z)^\top).
    \end{equation*}
    The Lipschitz constant of $\sigma(z)$ can be characterized as
    \begin{equation*}
        \ell_\sigma := \sup_{z\in\mathbb{R}^p}{\|J(z)\|} = \sup_{z\in\mathbb{R}^p}{\max_{j\in \mathcal{P}}{\rho_j(J(z))}},
    \end{equation*}
    where $\rho_j(J(z))$ denotes the $j$th eigenvalue of $J(z)$. Using the Gershgorin Circle Theorem \cite{varga2011gervsgorin}, we can derive the following upper bound for $\rho_j(J(z))$, $\forall j\in \mathcal{P}$:
    \begin{equation*}
        \begin{aligned}
            \rho_j(J(z)) &\leq \sum_{j\in \mathcal{P}}{|J_{ij}(z)|}\\
            &= \eta\sigma_j(z)\bigg[(1-\sigma_j(z)) + \sum_{i\neq j}{\sigma_i(z)}\bigg]\\
            &= 2\eta\sigma_j(z)(1-\sigma_j(z))\leq \frac{\eta}{2},
        \end{aligned}
        \end{equation*}
        where the second equality follows from 
        $\sum_{i\in \mathcal{P}}{\sigma_i(z)}=1$ and the last inequality holds since $a(1-a)\leq \frac{1}{4}$ for any $a\in \mathbb{R}$. This completes the proof.
\end{proof}
\begin{remark}
    Lemma \ref{softmax_lipschitz} establishes a smaller Lipschitz constant compared to earlier works, e.g., \cite{gao2017properties}, which verifies that the softmax function is $\eta$-Lipschitz.
\end{remark}
\begin{proposition}\label{uniqueness}
Let Assumptions \ref{f_strong_mono} and \ref{lipschitz_u} hold.
Then \eqref{vec_dyn_eq} has at most one equilibrium in 
$\mathcal{X}\times \mathcal{R}$ provided that
\begin{equation}\label{uniqueness_eq}
    \ell_M < \frac{2\mu_m}{\lambda\eta},
\end{equation}
where $\ell_M:=\max_{j\in \mathcal{P}}{\ell_j}$ and $\mu_m := \min_{j\in \mathcal{P}}{\mu_j}$.
\end{proposition}
\begin{proof}
    Denote $\sigma_u(x):=\sigma(u(x))$, where $\sigma$ is the softmax function defined in \eqref{softmax_func}.  From \eqref{nec_suff_exist_eq}, observe that any equilibrium of \eqref{vec_dyn_eq}
    satisfies
    \begin{equation*}\label{g_def}
        g(x):= f(x) - \lambda \sigma_u(x) = 0,
    \end{equation*}
    for some $x=x^{\mathrm{eq}}\in\mathcal{X}$. This equation has a unique solution if $g$ is strictly monotone over $\mathcal{X}$. That is, if 
    \begin{equation}\label{unique_monotone_cond}
        (x-x')^\top(g(x)-g(x')) > 0, \quad \forall x,x'\in\mathcal{X}.
    \end{equation}
    From Assumption \ref{f_strong_mono}, it follows that
    \begin{equation}\label{f_monotone}
        (x-x')^\top(f(x)-f(x')) \geq \mu_m\|x-x'\|^2, \quad \forall x,x'\in\mathcal{X}.
    \end{equation}
    Moreover, we have
\begin{equation}\label{sigma_lipschitz}
        \|\sigma_u(x) - \sigma_u(x')\|\leq \frac{\eta}{2}\|u(x) - u(x')\|\leq \frac{\eta}{2} \ell_M\|x-x'\|,
    \end{equation}
    where the inequalities follow from Lemma \ref{softmax_lipschitz} and Assumption \ref{lipschitz_u}, respectively.
    Combining \eqref{f_monotone} and \eqref{sigma_lipschitz}, we find that
    \begin{equation*}\label{g_monotone}
        \begin{aligned}
        &(x-x')^\top (g(x)-g(x'))\\&
            =(x-x')^\top (f(x)-f(x'))- \lambda (x-x')^\top (\sigma_u(x)-\sigma_u(x'))\\
            &\geq \mu_m\|x-x'\|^2 -\lambda \|x-x'\|\|\sigma_u(x)-\sigma_u(x')\|\\
            &\geq \mu_m\|x-x'\|^2 -\frac{\lambda\eta}{2} \ell_M\|x-x'\|^2.
        \end{aligned}
    \end{equation*}
    Consequently, \eqref{unique_monotone_cond} is satisfied if \eqref{uniqueness_eq} holds.
\end{proof}
\subsection{Asymptotic stability}
We now establish conditions under which 
the equilibrium
of the closed-loop system \eqref{vec_dyn_eq} is asymptotically stable.

\begin{proposition}\label{asymp_stability}
    Suppose Assumptions \ref{f_strong_mono} and \ref{lipschitz_u} hold. 
    Let $(x^\mathrm{eq},r^\mathrm{eq})\in\mathrm{int}(\mathcal{X})\times\mathcal{R}$ be an equilibrium of \eqref{vec_dyn_eq}.
    This equilibrium is unique and
    asymptotically stable if
\begin{equation}\label{stability_cond}
        \ell_M < \frac{2\mu_m}{\lambda\eta},
    \end{equation}
    where $\ell_M:=\max_{j\in \mathcal{P}}{\ell_j}$ and $\mu_m:=\min_{j\in \mathcal{P}}{\mu_j}$. 
\end{proposition}
\begin{proof}
     By Proposition \ref{uniqueness}, the equilibrium of the system is unique under the condition \eqref{stability_cond}. Now, consider the quadratic Lyapunov candidate
    \begin{equation*}
        V(x,r)=\frac{\alpha}{2}\|x-x^{\mathrm{eq}}\|^2 +\frac{1}{2}\|r-r^{\mathrm{eq}}\|^2, \quad \alpha>0.
    \end{equation*}
    Denote $\sigma_u(x):=\sigma(u(x))$, where $\sigma$ is the softmax function defined in \eqref{softmax_func}. By taking the time derivative of $V$ along the solutions $(x, r)\in \mathcal{X}\times\mathcal{R}$, we obtain
    \begin{equation*}
    \begin{aligned}
        \dot{V}&(x,r) = -\alpha(x-x^{\mathrm{eq}})^\top (f(x)-f(x^{\mathrm{eq}})) \\&+ \alpha\lambda(x-x^{\mathrm{eq}})^\top(r-r^{\mathrm{eq}}) - \|r-r^{\mathrm{eq}}\|^2 \\&+(r-r^{\mathrm{eq}})^\top(\sigma_u(x) - \sigma_u(x^{\mathrm{eq}}))\\
        &\leq -\alpha(x-x^{\mathrm{eq}})^\top M(x-x^{\mathrm{eq}}) + \alpha\lambda(x-x^{\mathrm{eq}})^\top(r-r^{\mathrm{eq}}) \\&+ (r-r^{\mathrm{eq}})^\top \Phi(x)(x-x^{\mathrm{eq}}) - \|r-r^{\mathrm{eq}}\|^2,
    \end{aligned}
    \end{equation*}
    where $M:=\mathrm{diag}(\mu_j)$ and 
    \begin{equation*}
        \Phi(x) := \int_0^1{\frac{\partial \sigma_u}{\partial x}(x^{\mathrm{eq}}+t(x-x^{\mathrm{eq}}))dt}.
    \end{equation*}
   Note that we have used the 
   vector-valued mean theorem to write the inequality.
   Therefore, $\dot{V}(x,r)<0$ if
    \begin{equation*}
        \begin{bmatrix}
            -2\alpha M&\alpha\lambda I + \Phi^\top(x)\\
            \alpha\lambda I + \Phi(x) & -2I
        \end{bmatrix}\prec 0.
    \end{equation*}
 By Schur complement, the above holds if and only if
    \begin{equation*}
        2\alpha M - \frac{1}{2}(\alpha I + \Phi(x))^\top(\alpha I + \Phi(x))\succ 0,
    \end{equation*}
    which is satisfied if
    \begin{equation}\label{stability_intermediate_cond}
        2\sqrt{\alpha \mu_m} > \|\alpha\lambda I + \Phi(x)\|.
    \end{equation}
    
    Define $z_t:=x^{\mathrm{eq}}+t(x-x^{\mathrm{eq}})$. We have
    \begin{equation*}
    \begin{aligned}
        \|\alpha\lambda I + \Phi(x)\| &\leq \alpha\lambda + \|\Phi(x)\|\\
        &\leq \alpha\lambda + \int_0^1{\bigg\|\frac{\partial \sigma_u}{\partial x}(z_t)\bigg\| dt}\\
        &\leq \alpha\lambda + \int_0^1{\bigg\|\frac{\partial \sigma}{\partial x}(u(z_t))\frac{\partial u}{\partial x}(z_t)\bigg\| dt}\\
        &\leq \alpha\lambda + \int_0^1{\bigg\|\frac{\partial \sigma}{\partial x}(u(z_t))\bigg\|\bigg\|\frac{\partial u}{\partial x}(z_t)\bigg\| dt}\\
        &\leq \alpha\lambda + \frac{\eta \ell_M}{2},
    \end{aligned}
    \end{equation*}
    where we used Assumption \ref{lipschitz_u} and Lemma \ref{softmax_lipschitz} to write the last inequality.
    Therefore, \eqref{stability_intermediate_cond} holds if 
\begin{equation*}
        2\sqrt{\alpha\mu_m} - \alpha\lambda > \frac{\eta \ell_M}{2}.
    \end{equation*}

The left-hand side is maximized by choosing $\alpha=\mu_m/\lambda^2$, so that we can obtain inequality \eqref{stability_cond}.
Hence, we conclude that, under the condition \eqref{stability_cond}, $\dot{V}(x,r)<0$, $\forall (x,r)\in \mathcal{X}\times\mathcal{R}$, $(x,r)\neq (x^\mathrm{eq},r^\mathrm{eq})$, which completes the proof.
\end{proof}

In the following discussion, we specify a subset of the free-flow region $\mathcal{X}\times\mathcal{R}$  which is positively invariant.  This subset serves a double purpose: the solutions initialized in this subset will not leave the free-flow region, thereby avoiding congestion, and, moreover, it provides an estimate of the region of attraction of the equilibrium.

\begin{lemma}\label{invariance_lemma}
Assume that \eqref{exist_cond} holds.
Then, the set $\mathcal{X}\times\mathcal{R}'$ where $\mathcal{R}':=\{r\in\mathcal{R}\ | \ r\leq\bar{f}/\lambda\}$ is positively invariant under the dynamics \eqref{vec_dyn_eq}. 
\end{lemma}

\begin{proof}
The compact set $\mathcal{X}\times\mathcal{R}'$ is positively invariant if the following conditions hold for each $j\in \mathcal{P}$:
\begin{enumerate}[i)]
    \item $\dot{r}_j \geq 0$ when $r_j=0$,\; $x_j\in \mathcal{X}_j$,
    \item $\dot{r}_j \leq 0$ when $r_j=\bar{f}_j/\lambda$,\; $x_j\in \mathcal{X}_j$,
    \item $\dot{x}_j \geq 0$ when  $x_j = 0$, $r\in\mathcal{R}'$,\;
    \item $\dot{x}_j \leq 0$ when  $x_j = \bar{x}_j$, $r\in\mathcal{R}'$.
\end{enumerate}
Conditions (i) and (iii) hold trivially by direct substitutions. By the dynamics \eqref{vec_dyn_ratio}, condition (ii) holds if
\begin{equation*}
    -\frac{\bar{f}_j}{\lambda} + \frac{e^{-\eta u_j(x_j)}}{\sum_{i\in\mathcal{P}}{e^{-\eta u_i(x_i)}} }\leq 0.
\end{equation*}
Bearing in mind that $\ubar{u}_j\leq u_j(x_j) \leq \bar{u}_j$, $\forall x_j\in \mathcal{X}_j$, 
the above inequality follows from \eqref{exist_cond}. Finally, (iv) holds if $\bar{f}_j\geq \lambda r_j$, which is satisfied since $r\in\mathcal{R}'$.
\end{proof}

\begin{corollary}
    Let Assumption \ref{f_strong_mono}, Assumption \ref{lipschitz_u}, inequality \eqref{exist_cond}, and inequality \eqref{stability_cond} hold. Then, any solution $(x,r)$ initialized in $\mathcal{X}\times\mathcal{R}'$, with $\mathcal{R}'$ defined as in Lemma \ref{invariance_lemma}, asymptotically converges to the unique equilibrium $(x^\mathrm{eq},r^\mathrm{eq})\in\mathcal{X}\times\mathcal{R}'$.
\end{corollary} 
\begin{proof}
     Assumption \ref{f_strong_mono} and inequality \eqref{exist_cond} ensure existence of an equilibrium. From \eqref{eq_eq}, we have
    \[
    r^\mathrm{eq} = \frac{f(x^\mathrm{eq})}{\lambda} \leq \frac{\bar{f}}{\lambda},
    \]
    demonstrating that the free-flow equilibrium of \eqref{vec_dyn_eq} belongs to the set $\mathcal{X}\times\mathcal{R}'$. Moreover, Assumptions \ref{f_strong_mono} and \ref{lipschitz_u} together with  inequality \eqref{stability_cond}, guarantee the uniqueness of the equilibrium $(x^\mathrm{eq},r^\mathrm{eq})\in\mathcal{X}\times\mathcal{R}'$. 
    Note that the set $\mathcal{X}\times\mathcal{R}'$ is positively invariant by Lemma \ref{invariance_lemma}.
    By the proof of Proposition \ref{asymp_stability}, we have $\dot{V}\leq 0$, $\forall (x,r)\in\mathcal{X}\times\mathcal{R}$, and $\dot V=0$ if and only if $(x,r)=(x^\mathrm{eq},r^\mathrm{eq})$.
    Hence, any solution of \eqref{vec_dyn_eq} that is initialized in $\mathcal{X}\times\mathcal{R}'$, remains in $\mathcal{X}\times\mathcal{R}'$ and asymptotically converges to  $(x^\mathrm{eq},r^\mathrm{eq})\in\mathcal{X}\times\mathcal{R}'$, where the latter follows from LaSalle's invariance principle. 
\end{proof}

\medskip{}
Combining the results of this section, 
we can now characterize the class of traffic information $\mathcal{U}$ specified in Problem \ref{problem}. 
In particular, let Assumption~\ref{f_strong_mono} hold on the outflow functions. Then, 
$\mathcal{U}$ is given by
\begin{equation*}
\begin{aligned}
    \mathcal{U}=\Bigg\{ u:\mathcal{X}\rightarrow\mathbb{R}_{\geq 0}^p  \bigg|  \eqref{exist_cond}, 
\bigg\lvert\frac{\partial u_j}{\partial x_j}\bigg\rvert < \frac{2\mu_m}{\eta\lambda},  \forall x_j\in\mathcal{X}_j,   \forall j\in\mathcal{P}
\Bigg\}.
\end{aligned}
\end{equation*}
For any $u\in\mathcal{U}$, the closed-loop system \eqref{vec_dyn_eq} admits a unique free-flow equilibrium, which is asymptotically stable. 
\section{Illustrative Example}\label{example}
Building on the theoretical results established in Section \ref{condition_sec}, here we will numerically demonstrate how a carefully designed information signal stabilizes the traffic dynamics in the free-flow regime, whereas routing based on real travel-time information may lead to congestion.
\subsection{Information design}
Consider the traffic model given in \eqref{vec_dyn_eq}, with outflow functions $f_j(x_j)=\min{\{\mu_jx_j,\ \bar{f}_j\}}$ and travel times given by the BPR function \cite{sheffi1985urban} 

\begin{equation}\label{bpr}
    \tau_j(x_j) = t_j^0\left(1+\theta\left(\frac{x_j}{\bar{x}_j}\right)^\delta\right),
\end{equation}
where $t_j^0>0$ is the free-flow travel time on path $j$. In \eqref{bpr}, $\theta>0$ and $\delta>0$ are empirical parameters, commonly chosen as $\theta=0.15$ and $\delta=4$.

To design the traffic information, we consider a class of affine functions
\begin{equation*}
    u(x) = A x + b, \qquad A := \mathrm{diag}(a_j), \ b := \mathrm{col}(b_j).
\end{equation*}
The parameters $A$, $b$, together with a target equilibrium $x^*$, are chosen following the optimization problem:
\begin{equation}\label{select_u}
    \begin{aligned}
        \min_{x^*\in\mathcal{X}, A, b}\; & \sum_{j\in \mathcal{P}} 
        \phi_j(x_j^*, a_j, b_j)\\
        \text{s.t.} \quad & Ax+b\in\mathcal{U}, \ f(x^*)=\lambda\frac{e^{-\eta u(x^*)}}{\mathds{1}^\top e^{-\eta u(x^*)}}
    \end{aligned}
\end{equation}
where
\begin{equation}\label{cost_func}
\begin{aligned}
    \phi_j(x_j^*, a_j, b_j)&:=       
        f_j(x_j^*)\tau_j(x_j^*) 
        \\&+ \gamma \int_0^{\bar{x}_j}{\big(a_j y + b_j - \tau_j(y)\big)^2} dy.
\end{aligned}
\end{equation}
The parameter $\gamma \geq 0$ in \eqref{cost_func} balances two objectives: the networks efficiency formulated as the  aggregate travel time (first term), and the credibility of the traffic information provided to the users, which is formulated as the squared distance of the information from the actual travel time in the free-flow region (second term).
The constraint $u \in \mathcal{U}$ ensures that the traffic information $u(\cdot)$ results in a unique asymptotically stable free-flow equilibrium $(x^\mathrm{eq},r^\mathrm{eq})$, while the additional equality constraint ensures that $x^\mathrm{eq}=x^*$.
\subsection{Numerical results}
We present a numerical analysis of the proposed model under the following setting. The inflow is fixed at $\lambda = 1$, and the network consists of $p = 5$ parallel paths with critical density $\bar{x} = (0.15, 0.15, 0.175, 0.2, 0.2)$ and $\mu = (2, 2, 3, 2.5, 4)$. The parameters in \eqref{bpr} are chosen as $t^0 = (8, 6, 5, 5, 2)$, $\theta = 1.5$, and $\delta = 2$.

When users are provided with actual travel-time information, i.e., when $u(x) = \tau(x)$ in system \eqref{vec_dyn_eq}, the dynamics may converge to an equilibrium where one or more paths operate in the congested regime. This effect becomes more pronounced as the responsiveness rate $\eta$ increases, since drivers react more strongly to small differences in the perceived travel times, leading to concentrated routing choices and potential overload of certain paths. As illustrated in the top panel of Figure \ref{tt_eq_fig}, for $\eta \geq 7.94$, the equilibrium density on link 5 exceeds its critical threshold ($x_5^{\mathrm{eq}} > \bar{x}_5$), meaning that link 5 transitions from free-flow to the congested regime. Conversely, the bottom panel of Fig. \ref{tt_eq_fig} shows the equilibrium link densities when users are provided with the affine information signal $u(x)$ obtained by \eqref{select_u} with $\gamma=0$. This designed information steers the system toward the socially optimal free-flow equilibrium, where all equilibrium densities remain in the free-flow regime.
\begin{figure}[t]
    \centering
    \includegraphics[width=.85\linewidth]{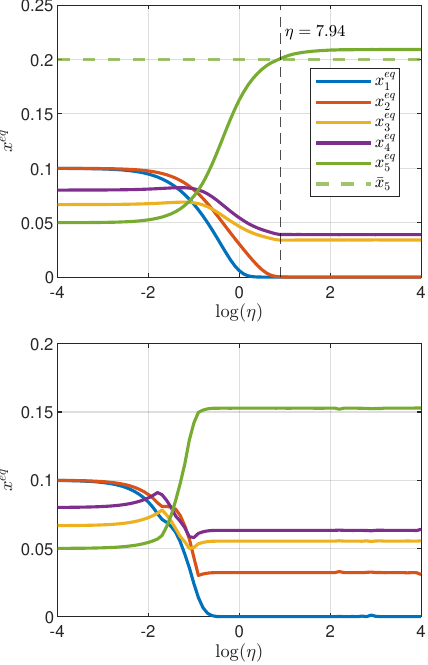}
    \caption{Equilibrium link densities with respect to $\eta$ when users are provided with real travel-time information $\tau(x)$ (top), and designed information $u(x)$ with $\gamma=0$ (bottom).}
    \label{tt_eq_fig}
\end{figure}

While fixing 
$\gamma=0$ maximizes the efficiency of system \eqref{vec_dyn_eq} at equilibrium, the resulting information signal may be inconsistent with the travel times experienced by users. Such disparity raises credibility concerns, and drivers may be reluctant to make routing decisions based on the provided information.  
This is accommodated in the optimization \eqref{select_u}, where $\gamma$ penalizes the aforementioned disparity. The effect of $\gamma$ on the system’s efficiency, measured by aggregate travel time at the equilibrium, and on the credibility of the information, quantified by the error $\|\tau(x(t)) - u(x(t))\|$, are investigated in Figures \ref{total_travel_time_vs_eta} and \ref{tau_u_error}, respectively. As shown in the figures, increasing $\gamma=0$ to $\gamma=0.1$ improves the credibility of the provided information substantially while practically maintaining the optimal efficiency of the system.
Increasing $\gamma$ further improves information credibility but reduces overall travel-time efficiency.
\begin{figure}[t]
    \centering
    \includegraphics[width=.8\linewidth]{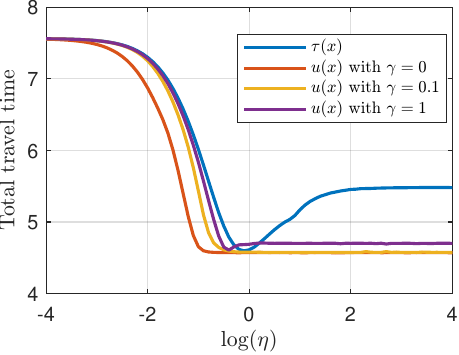}
    \caption{Comparison of total travel times experienced by the drivers at equilibrium, when provided with different travel-time information.}
    \label{total_travel_time_vs_eta}
\end{figure}
\begin{figure}[t]
    \centering
    \includegraphics[width=.8\linewidth]{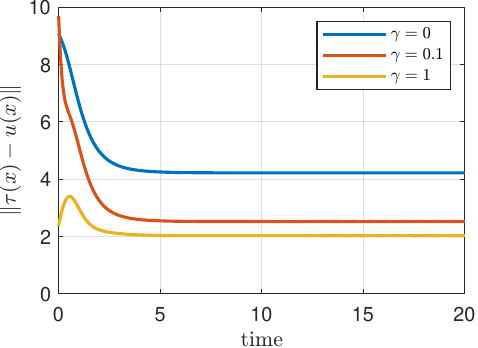}
    \caption{Credibility of the designed information $u(x)$ with different values of $\gamma$, quantified by $\|\tau(x)-u(x)\|$.}
    \label{tau_u_error}
\end{figure}
\par
Solving \eqref{select_u} for $\gamma=0.1$ and $\eta=20$, 
the parameters of $u(x)$ are determined as $a=(0.2,-0.19,0.2,0.2,0)$ and $b=(6.84,6.13,6.05,6.06,6)$, and  the corresponding equilibrium point is $x^*=(0, 0.026, 0.056, 0.063, 0.156)$ and $r^*=(0,	0.052,	0.167,	0.158,	0.623)$. The trajectory of the densities and routing ratios, starting from an arbitrary point $(x^0,r^0)\in\mathcal{X}\times\mathcal{R}'$, is shown in Figure \ref{trajectory_fig}. 
Furthermore, the positive invariance of the set $\mathcal{X}\times\mathcal{R}'$ is illustrated through Figure \ref{invariance_fig}, where $100$ trajectories of $x(t)/\bar{x}$ and $r(t)/\bar{f}$, initialized arbitrarily in $\mathcal{X}\times \mathcal{R}'$, are presented. We can see that $x(t)/\bar{x}\in [0,1]$ and $r(t)/\bar{f}\in [0,1/\lambda]$ for all $t\geq 0$, demonstrating that solutions initialized in the set $\mathcal{X}\times \mathcal{R}'$ remain within it.

    \begin{figure}[t]
    \centering
    \includegraphics[width=\linewidth]{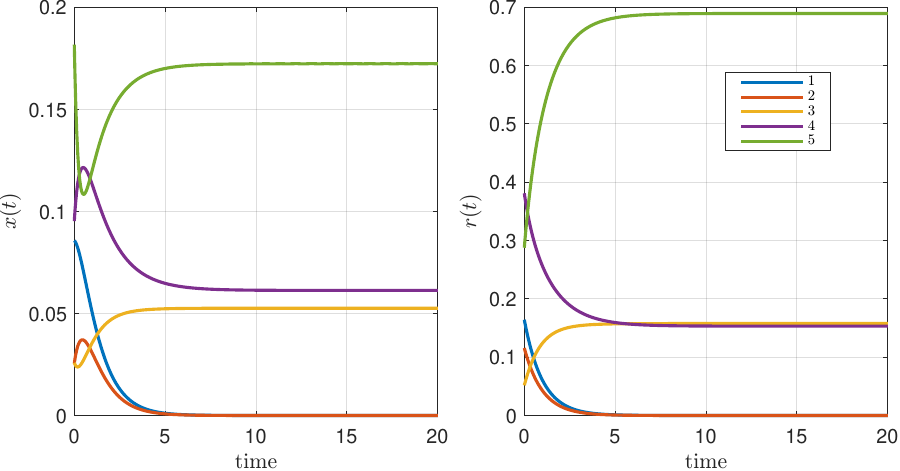}
    \caption{Trajectory of link densities and routing ratios under the designed $u(x)$ with $\eta=20$ and $\gamma=0.1$.}
    \label{trajectory_fig}
\end{figure}

\begin{figure}[t]
    \centering
    \includegraphics[width=\linewidth]{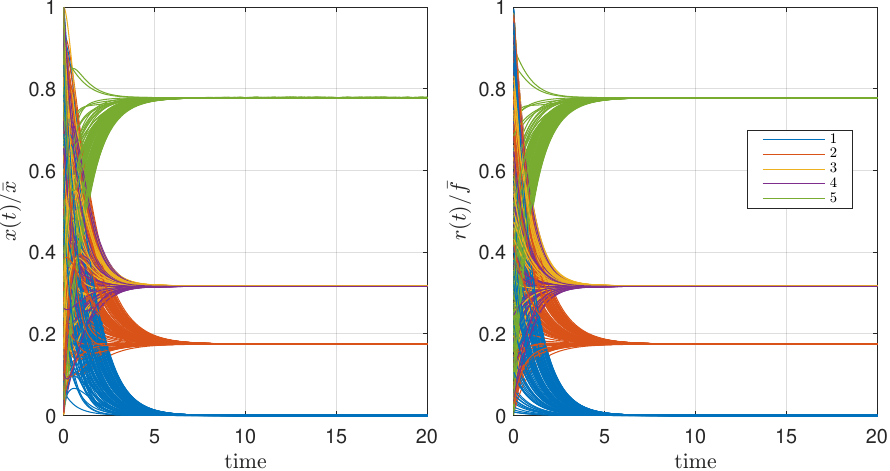}
    \caption{$100$ trajectories of $x(t)/\bar{x}$ and $r(t)/\bar{f}$ starting from $x^0\in\mathcal{X}\times \mathcal{R}'$ under the designed $u(x)$ with $\eta=20$ and $\gamma=0.1$.}
    \label{invariance_fig}
\end{figure}
\section{Conclusion}\label{conclusion_sec}
This paper studied a parallel-path traffic network with joint density–routing dynamics evolving at the same timescale, where routing behavior follows logit dynamics influenced by an announced traffic information. We derived sufficient conditions on the traffic information signals such that there exists a unique free-flow equilibrium that is asymptotically stable. Moreover, we identified a positively invariant set such that all trajectories starting in the free-flow regime remain within that regime, avoiding congestion at all times. In addition to the analytical results, the design of traffic information, satisfying the desired conditions, and its effectiveness was investigated through a numerical example where we demonstrated the efficiency and credibility of the designed traffic information.

A key direction for future research is to extend the framework to general network topologies. In addition, the control of such networks through information design should be explored in both the free-flow and congested regimes.

\bibliographystyle{IEEEtran}
\bibliography{DraftRef}

\end{document}